\documentclass[english,11pt]{article}
\usepackage[papersize={8.5in,11in},margin=1in]{geometry}
\usepackage{framed}
\usepackage{amsthm}
\usepackage{amsmath}
\usepackage{amssymb}
\usepackage{mathrsfs}
\usepackage{xspace}
\usepackage{booktabs}
\usepackage{aliascnt}
\usepackage{natbib}
\usepackage{balance}
\usepackage{enumitem}

\usepackage[pagebackref=true,colorlinks]{hyperref}
\hypersetup{linkcolor=red,filecolor=red,citecolor=red,urlcolor=red}

\newtheorem{theorem}{Theorem}[section]

\newaliascnt{lemma}{theorem}
\newaliascnt{assumption}{theorem}
\newaliascnt{proposition}{theorem}
\newaliascnt{corollary}{theorem}
\newaliascnt{claim}{theorem}
\newaliascnt{observation}{theorem}
\newaliascnt{definition}{theorem}
\newaliascnt{fact}{theorem}
\newaliascnt{statement}{theorem}
\newaliascnt{mechanism}{theorem}

\newtheorem{lemma}[lemma]{Lemma}

\newtheorem{definition}[definition]{Definition}

\aliascntresetthe{lemma}
\aliascntresetthe{assumption}
\aliascntresetthe{proposition}
\aliascntresetthe{corollary}
\aliascntresetthe{claim}
\aliascntresetthe{observation}
\aliascntresetthe{definition}
\aliascntresetthe{fact}
\aliascntresetthe{statement}
\aliascntresetthe{mechanism}

\newcommand{\argmax}{\operatornamewithlimits{argmax}}

\newcommand\R{\mathbb{R}}

\newcommand\Rev{\textsc{Rev}}

\newcommand\boldeta{\boldsymbol{\eta}}
\newcommand\boldlambda{\boldsymbol{\lambda}}
\newcommand\boldmu{\boldsymbol{\mu}}

\DeclareMathOperator*{\sign}{\mathrm sgn}
\DeclareMathOperator*{\E}{\mathbb E}

\newcommand{\Infer}{\Longrightarrow}

\newcommand{\ud}{\mathrm{d}}

\newcommand{\cu}{U}

\newcommand{\V}{\mathcal{V}}
\newcommand{\x}{\boldsymbol{x}}

\renewcommand{\b}{\boldsymbol{b}}
\newcommand{\Zero}{\boldsymbol{0}}

\renewcommand{\v}{{\boldsymbol{v}}}
\newcommand{\p}{\boldsymbol{p}}

\makeatletter
\renewcommand{\paragraph}{%
  \@startsection{paragraph}{4}%
  {\z@}{1.75ex \@plus 1ex \@minus .2ex}{-1em}%
  {\normalfont\normalsize\bfseries}%
}
\makeatother

\newcommand{\ubar}{\underline}
\newcommand{\short}[2]{#1}

\begin{document}

\title{Optimal Dynamic Auctions are Virtual Welfare Maximizers\thanks{The
authors thank the anonymous reviewers for their helpful comments.}}
\author{Vahab Mirrokni\textsuperscript{$\dag$}
\and Renato Paes Leme\textsuperscript{$\dag$}
\and Pingzhong Tang\textsuperscript{$\ddag$}
\and Song Zuo\textsuperscript{$\dag$}\\
\textsuperscript{$\dag$}Google Research,
\texttt{\{mirrokni,renatoppl,szuo\}@google.com}\\
\textsuperscript{$\ddag$}Tsinghua University, \texttt{kenshinping@gmail.com}}
\date{}
\maketitle
\begin{abstract}
%

We are interested in the setting where a seller sells sequentially arriving
items, one per period, via a {\em dynamic auction}. At the beginning of each
period, each buyer draws a private valuation for the item to be sold in that
period and this valuation is independent across buyers and periods. The auction
can be dynamic in the sense that the auction at period $t$ can be conditional on
the bids in that period and all previous periods, subject to certain
appropriately defined {\em incentive compatible} and {\em individually rational}
conditions. Perhaps not surprisingly, the revenue optimal dynamic auctions are
computationally hard to find and existing literatures that aim to approximate
the optimal auctions are all based on solving complex dynamic programs. It
remains largely open on the structural interpretability of the optimal dynamic
auctions.

In this paper, we show that any optimal dynamic auction is a virtual welfare
maximizer subject to some monotone allocation constraints. In particular, the
explicit definition of the virtual value function above arises naturally from
the primal-dual analysis by relaxing the monotone constraints. We further
develop an ironing technique that gets rid of the monotone allocation
constraints. Quite different from Myerson's ironing approach, our technique is
more technically involved due to the interdependence of the virtual value
functions across buyers. We nevertheless show that ironing can be done
approximately and efficiently, which in turn leads to a Fully Polynomial Time
Approximation Scheme of the optimal dynamic auction.

\end{abstract}
\clearpage

\section{Introduction}\label{sec:intro}

Recently, the problem of designing multi-period dynamic mechanisms has been shown
to be both theoretically challenging and practically important. In particular,
when running a sequence of repeated auctions on online advertising platforms,
using dynamic auctions optimized across different time periods could potentially
bring significant gains both in terms of revenue and social welfare. The power
of dynamic mechanisms has been investigated by a number of recent papers
\citep{bergemann2002information,parkes2004mdp,cavallo2008efficiency,athey2013efficient,kakade2013optimal,pai2013optimal,pavan2014dynamic,devanur2014perfect,balseiro2016dynamic,chawla2016simple,BergemannCastroWeintraub2017EC,BLMPZ17,BalseiroMirRokniPaesLeme2017EC,lobel2017dynamic,shen2018expost,balseiro2019dynamic}.
We would also refer to \citet{bergemann2011dynamic} and \citet{bergemann2017dynamic} for
comprehensive surveys on the subject.

In particular, we consider a setting where a seller repeatedly interacts with a
set of buyers and sells one item per period. The value of each buyer for the
item at each period are independently drawn from commonly known prior
distributions (no need to be identical) at the beginning of that period. The
seller is allowed to sell the item of each period not only depending on the bids
submitted in the current period, but also the histories, i.e., all the bids
submitted in past periods. In the meanwhile, the seller must guarantee the
dynamic auctions to be {\em ex-post individual rational} --- the cumulative
utility for each buyer is always positive, and {\em dynamic incentive compatible}
--- bidding truthfully (i.e., submitting private values as bids) to the auction
is optimal for each buyer by taking into consideration the effect of current
bids on future outcomes.

Even though dynamic mechanisms can be much more effective in maximizing revenue
and social welfare \citep{jackson2007overcoming,papadimitriou2016complexity},
they have not been widely adopted in practice. The main issue of implementing
dynamic mechanisms in practice is their high complexity. The complexity induced
by the exponentially growing design space makes it difficult to solve or even to
describe such mechanisms.

A series of recent work has made progress to resolve the complexity issue
described above. For example, \citet{ashlagi2016sequential} and
\citet{MirrokniLTZ17} show that it is enough for the optimal
dynamic auction to depend on scalar summaries of histories instead of the
full history.
However, the approximately optimal mechanisms described in these papers are
solutions of complex dynamic programs that are written into a large table. It is
therefore very difficult to understand the structure of these mechanisms. It
remains open whether there is an intuitive structural characterization of the
optimal dynamic auction.

In this paper, we show that the exact optimal dynamic auction has a very simple
structural interpretation: the optimal dynamic auction in each period is a
second price auction on a certain appropriately defined virtual value
space.\footnote{A virtual value function is a map from buyer value space to real
numbers. A virtual value is the corresponding real number of some private value
of the buyer.} More specifically, such virtual values (before
ironing\footnote{Informally, the ironing of a virtual value function is an
operation that maps the virtual value function to another virtual value function
(we call an ironed virtual value function) while preserving the expectation. For
the formal definition, see \autoref{def:ironing}.}\footnote{Some recent works on
virtual value and ironing \citep{elkind2007designing,roughgarden2016ironing}.})
are quite similar to Myerson's virtual value \citep{myerson1981optimal}, i.e.,
they have the form of linear combinations of private values and Myerson's
virtual values (before ironing). However, just like Myerson's auction, to make
the virtual welfare\footnote{Virtual welfare is the sum of virtual values of the
buyers who get the item. For expected virtual welfare, the expectation is taken
over the randomness from the allocation rules.} maximizing allocation rule
monotone, one need to first iron the virtual values. Unlike the ironing in
Myerson's auction, the ironing step in our case is interdependent across the
values of different buyers. In other words, one's virtual value after ironing
not only depends on his/her own value, but also on other buyers' values.
Although the ironing step here is not as simple as the ironing in Myerson's
auction, its computation is still efficient for constant many buyer cases.
Moreover, we provide a Fully Polynomial Time Approximation Scheme to compute the
virtual values for any period of the dynamic auction given the histories so far,
and such virtual values induce a dynamic auction with revenue arbitrarily close
to the optimal.

\paragraph{Techniques}
There are two main techniques used in our analysis for optimal dynamic auctions.
The first is the so-called {\em bank account mechanisms} which are a subset of
dynamic auctions with simple structures and can achieve the optimal revenue
of all the dynamic auctions
\citep{mirrokni2016dynamic,mirrokni2016optimal,mirrokni2018dynamic,MirrokniLTZ17}.
Briefly speaking, a bank account mechanism keeps a state for each buyer as the
summary of the history of each buyer, which is a scalar called {\em balance}.
Each period depends on the previous periods through the vector of buyer balances.
With the bank account framework, the designer only needs to specify
single-period auctions that are single-period incentive compatible together with
a valid balance update policy. In other words, the design of an entire dynamic
auction breaks into the design of a series of single-period auctions and a
balance update policy. The decomposition greatly simplifies the problem and
enables clean mathematical programs for each period.

The second is a primal-dual analysis and a sensitivity analysis of the
parametric mathematical programs. In fact, primal-dual analysis is commonly used
in economic studies [\citealp[Chapter~5]{nisan2007algorithmic},
\citealp{daskalakis2013mechanism,cai2016duality,cai2017simple,daskalakis2017strong}].
In particular, for our problem, we can prove that these programs are convex and
satisfy the Slater's conditions. Hence the solution is optimal if and only if
the KKT conditions are satisfied. From these conditions, we show that the
auction in each period maximizes some virtual welfare and we also discover the
exact form of the virtual values with partial relaxation on the monotonicity
constraints of allocation. Furthermore, we show that adding these relaxed
constraints back to the program indeed corresponds to ironing the virtual values.
As we mentioned previously, the ironing step here is interdependent across the
values of the buyers and hence different from the ironing step in Myerson's
auction. To resolve this difficulty, we show a method to algorithmically
accomplish the ironing step.

\section{Preliminaries}\label{sec:prelim}

  We study a setting where a seller repeatedly interacts with $k$ buyers
  selling one item per period over $T$ periods. The value of each buyer
  $i \in [k]$ for the item in period $t \in [T]$ is $v^t_i \in \V$. If $x^t_i
  \in [0, 1]$ represents the probability that buyer $i$ is allocated the item
  in period $t$, his utility is $v^t_i \cdot x^t_i$.
  Throughout this paper, (i) we use subscripts as the indices of buyers, bold
  fonts for vectors of all the buyers (i.e., $\mathbf{a} = (a_1, \ldots, a_k)$),
  and subscript $-i$ for the vector except the $i$-th element (i.e., $a_{-i} =
  (a_1, \ldots, a_{i-1}, a_{i+1}, \ldots, a_k)$); (ii) we use superscripts as
  the indices of periods and $a^{1..t}$ to denote the sequence $a^1, \ldots,
  a^t$.

  The values $v_i^t$ are assumed to be drawn from independent distributions
  $F_i^t$.\footnote{In practice, it is still fair to assume that the values of
  the buyers (advertisers) are independent conditional on each particular inventory
  and cookie.}\footnote{If the distributions are not independent, then the weak
  version of truthfulness still holds, i.e., truthful if all other buyers bid
  truthfully.} The distributions $F_i^t$ are assumed to be common knowledge but
  the realizations of the random variables are initially unknown for both the
  buyers and the designer. At each period, the following events happen:
  \begin{enumerate}
    \itemsep0em
    \item each buyer $i$ learns his value $v_i^t \sim F_i^t$;
    \item each buyer $i$ reports value $\hat v_i^t$ to the designer;
    \item the designer implements an outcome $\x^t \in [0, 1]^k$ and charges the
          buyers $\p^t$;
    \item each buyer accrues utility $u_i^t = v_i^t \cdot x_i^t - p_i^t$.
  \end{enumerate}
  A dynamic mechanism can then be described in terms of a pair of maps for each
  period, which associate the history of reports $\hat \v^{1..t} = (\hat \v^1,
  \hat \v^2, \hdots, \hat \v^t)$ to an outcome $\x^t$ and payment $\p^t$:
  \begin{align*}
    \text{Outcome:}~\x^t: \V^{kt} \rightarrow [0, 1]^k \qquad
    \text{Payment:}~\p^t: \V^{kt} \rightarrow \R^k.
  \end{align*}
  Therefore we can define: $u_i^t(v_i^t; \hat{\v}^{1..t}) =
  v_i^t \cdot x_i^t(\hat{\v}^{1..t}) - p_i^t(\hat{\v}^{1..t})$.

  \subsection{Dynamic Incentive Compatibility}

    A mechanism is incentive compatible if it provides incentives for buyers to
    reveal their true types in each iteration. Such conditions for dynamic
    mechanisms can be easily defined by backward induction: in the last period,
    regardless of the history so far and other buyers' reports, it should be
    incentive compatible for each buyer to report his true value. This corresponds
    to the usual notion of incentive compatibility in (static) mechanism design:
    \begin{align*}
      \textstyle
      v_i^T = \argmax_{\hat v^T_i} u^T_i(v^T_i; \hat \v^{1..T})
    \end{align*}
    for all $i, \hat \v^{1..T-1}, \hat v_{-i}^T, v^T_i$. To simplify notations, from
    now on we will omit the `for-all' quantification and assume all expressions are
    quantified as `for-all' in its free variables.
    For the next-to-last-period, it should be incentive compatible for the buyer
    to report his true value given that he will report his true value in the
    following period:
    \begin{align*}
      \textstyle
        v_i^{T-1} =
        \argmax_{\hat v_i^{T-1}} u_i^{T-1}(v_i^{T-1}; \hat \v^{1..T-1})
        + \E_{v^T_i} [u^T_i(v^T_i; \hat \v^{1..T-1}, v^T_i, \hat v_{-i}^T)].
    \end{align*}
    Proceeding by backward induction\short{}{ for all periods}, we require that:
    \begin{align}\label{eq:dic}\tag{DIC}
      \textstyle
      v^t_i = \argmax_{\hat v^t_i} u^t_i(v^t_i; \hat \v^{1..t}) +
        \cu^t_i(\hat v_i^{1..t} | \hat v_{-i}^{1..T})
    \end{align}
    where the second term is the \emph{continuation utility}, i.e., the expected
    utility obtained from the subsequent periods of the mechanism assuming the buyer
    reports truthfully:
    \begin{align*}
      \textstyle
      \cu^t_i(\hat v_i^{1..t} | \hat v_{-i}^{1..T}) :=
        \E_{v_i^{t+1..T}}\left[\sum_{\tau=t+1}^T
          u_i^\tau(v_i^\tau; \hat \v^{1..t}, v_i^{t+1..\tau}, \hat v_{-i}^{t+1..\tau})\right]
    \end{align*}
    A well-known fact in dynamic mechanism design is that \ref{eq:dic}
    implies that buyer $i$'s expected overall utility $U_i^0(\hat v_{-i}^{1..T})$
    is maximized by reporting truthfully in each period.

  \subsection{Ex-Post Individual Rationality}

    Another desirable constraint is ex-post individual rationality which says that
    a buyer should derive non-negative utility from the mechanism for every
    realization of the values:
    \begin{align}\label{eq:epir}\tag{eP-IR}
      \textstyle
      \sum_{t=1}^T u^t_i(v^t_i; \v^{1..t}) \geq 0
    \end{align}
    We focus on the problem of maximizing revenue subject to \ref{eq:dic},
    \ref{eq:epir}, and feasibility constraints:
    \begin{align*}
      \max &\textstyle \quad \Rev = \E[\sum_{t=1}^T \sum_{i=1}^k p^t_i(\v^{1..t})]  \\
      \text{s.t.} & \textstyle \quad \text{\eqref{eq:dic}, \eqref{eq:epir},
        and feasibility:}~\sum_{i = 1}^k x^t_i(\v^{1..t}) \leq 1
    \end{align*}

  \subsection{Bank Account Mechanisms}\label{subsec:bank}

    The space of mechanisms satisfying \ref{eq:dic} and \ref{eq:epir} is very
    broad and unstructured. We restrict our attention in this section to a
    subclass of dynamic mechanisms introduced by \citet{mirrokni2016dynamic}
    called \emph{bank account mechanisms}. The mechanisms are simple, dynamic
    incentive compatible by design and have the following notable features:
    \begin{lemma}[\cite{mirrokni2016nonclairvoyant}]\label{lem:bankopt}
      Given any dynamic mechanism satisfying \ref{eq:dic} and \ref{eq:epir},
      there exists a bank account mechanism with at least the same revenue and
      at least the same welfare.

      In particular, for any given setting, there is a revenue-optimal mechanism
      in the form of a bank account mechanism.
    \end{lemma}

    Bank account mechanisms keep a state for each buyer, which is a scalar
    called \emph{balance}. Each period depends on the previous periods through
    the vector of buyer balances. Another main  feature is that in this
    framework, the designer needs to specify single-period auctions that are
    single-period incentive compatible together with a valid balance update
    policy. That is, once a valid balance update policy is in place, all the
    designer needs to worry about are single-period incentive compatibility
    constraints.

    A bank account mechanism $B$ is defined in terms of the following functions
    for each period:
    \begin{itemize}
      \itemsep0em
      \item
        A static single-period auction $\x^{B,t}(\v^t, \b)$, $\p^{B,t}(\v^t, \b)$
        parameterized by a balance vector $\b \in \R^k_+$ that is (single-period)
        \emph{incentive-compatible} for each $\b$, i.e.:
        \begin{align}\label{eq:ic}\tag{IC}
          v^t_i \cdot x^{B,t}_i(\v^t, \b) - p^{B,t}_i(\v^t, \b)
            \geq v^t_i \cdot x^{B,t}_i(\hat v^t_i, v_{-i}^t, \b)
                     - p^{B,t}_i(\hat v^t_i, v_{-i}^t, \b)
        \end{align}
      \item
        Note that we do not require the mechanism to be (single-period)
        individually rational. We also require the utility of the buyer to be
        \emph{balance independent in expectation}, i.e., that:
        \begin{align}\label{eq:bi}\tag{BI}
          \E_{v^t_i \sim F^t_i}\left[v^t_i \cdot x^{B,t}_i(\v^t, \b)
                                      - p^{B,t}_i(\v^t, \b)\right]
          \text{ is a non-negative constant not depending on } \b
        \end{align}
      \item
        A balance update policy $\b^{B,t}(\v^t, \b)$ which maps the previous balances and
        the reports to the current balances, satisfying the following
        \emph{balance update} conditions:
        \begin{align}\label{eq:bu}\tag{BU}
            0 \leq
            b_i^{B,t}(\v^t, \b) \leq b^i + v^t_i \cdot x^{B,t}_i(\v^t, \b) -
              p^{B,t}_i(\v^t, \b)
        \end{align}
        Given the balance update functions, we can define $\b^t : \V^t \rightarrow
        \R^k_+$ recursively as:
        \begin{align*}
            \b^0 = \Zero \quad \text{ and } \quad
            \b^1(\v^1) = \b^{B,1}(\v^1, \Zero) \quad \text{and} \quad \b^t(\v^{1..t}) =
            \b^{B,t}(\v^t, \b^{B,t-1}(\v^{1..t-1}))
        \end{align*}
        which allows us to define a dynamic mechanism in the standard sense as:
        \begin{align*}
          \x^t(\v^{1..t}) = \x^{B,t}(\v^t, \b^{t-1}),~~
          \p^t(\v^{1..t}) = \p^{B,t}(\v^t, \b^{t-1}).
        \end{align*}
    \end{itemize}

    In what follows we will abuse notations by dropping the superscript $B$ and
    refer to $\x^t(\v^{1..t})$ and $\x^t(\v^t, \b^{t-1})$ interchangeably. One
    important theorem from previous studies is that any bank account mechanism
    satisfies stronger notions of \ref{eq:dic} and \ref{eq:epir}.
    \begin{lemma}[\cite{mirrokni2016nonclairvoyant}]\label{lem:bank_dic}
      Any bank account mechanism satisfying \ref{eq:ic}, \ref{eq:bi}, and
      \ref{eq:bu} is \short{\ref{eq:dic} and \ref{eq:epir}}{dynamic incentive
      compatible (\ref{eq:dic}) and ex-post individually rational
      (\ref{eq:epir})}.
    \end{lemma}

\section{The Structure of Optimal Bank Account Mechanisms}\label{sec:program}

  In this section, we show our first main result that identifies the underlying
  structure of the optimal bank account mechanism (hence also the optimal
  dynamic auctions). Interestingly, the dynamic auctions with the structure can
  be interpreted as an ironed virtual welfare maximizing auction, where the
  ironed virtual value is defined as follows.

  \begin{definition}[Ironing]\label{def:ironing}
    The ironing operation on the virtual value functions of all buyers transfers
    these virtual value functions into some other virtual value functions
    (called ironed virtual value functions) while preserving the following
    properties:
    \begin{itemize}
      \itemsep0em
      \item {\em Monotonicity}: The allocation rule that maximizes the ironed
            virtual welfare is monotone, i.e., for each buyer $i$, the
            allocation probability is weakly increasing in $v_i$ for any fixed
            $v_{-i}$.
      \item {\em Limited transfer}: The expectation of virtual value conditional
            on each allocation equivalence class is unchanged. An allocation
            equivalence class of buyer $i$ with any given $v_{-i}$ is a maximal
            subset of his/her private values where the allocation probability to
            buyer $i$ is a constant. Again, the allocation rule here is the one
            that maximizes the ironed virtual welfare.
      \item {\em First order dominance}: The ironed virtual value is first order
            dominated by the virtual value before ironing.
    \end{itemize}
  \end{definition}

  Informally, we have the following theorem:
  \begin{theorem}[Informal]\label{thm:informal}
    Any revenue optimal bank account mechanism is maximizing some virtual value
    (after ironing) for each period.
  \end{theorem}

  We will restate the formal version of this theorem after introducing all
  necessary notations (\autoref{thm:formal}). In particular, the specific form
  of the virtual value (before ironing) will be provided.

  \subsection{Formalizing the subprogram for each period}

  We start with the subproblem of optimizing period $t$ while all other periods
  are fixed. In particular, the problem can be formalized as a convex program.
  \begin{lemma}[Convex program]\label{lem:program}
    For each period $t$, any valid balance $\b^t$, and fixed expected utility
    for period $t$, if mechanism is fixed for the remaining periods ($t+1$ to
    $T$), then the optimal auction for period $t$ can be solved via a convex
    program.
  \end{lemma}

  To proof this statement smoothly, we show how to formalize the program step by
  step in the rest of this subsection rather than putting everything into a
  single proof environment. For the ease of presentation, we will hide the
  superscript ${}^t$ while focusing on a single period.

  Let $\b$ be the bank account balance given at the beginning of this period and
  \begin{align}\label{eq:xibi}
    \textstyle
    \xi_i(v_{-i}) := \E_{v_i}[v_i \cdot x(\v, \b) - p(\v, \b)]
  \end{align}
  be the expected utility of buyer $i$ in this period. Note that by \eqref{eq:bi},
  it is independent of the bank account balance $\b$, while it could be
  different for different bids from other buyers, i.e., $v_{-i}$. In fact, the
  $\xi$ are some parameters we will need to determine later based on the
  distribution of $\b$ induced by the auctions in each period.

  \paragraph{Objective}
  The expected revenue acquired within this period can be computed by the
  expected social welfare minus the expected utility of all the buyers,
  \begin{align*}
    \text{\sc Period}\Rev(\b, \xi)
      = \E\nolimits_\v\left[\sum\nolimits_i(v_i \cdot x_i(\v, \b) - \xi_i(v_{-i}))\right].
  \end{align*}
  Besides the expected revenue from the current period, we also need to include
  the expected revenue for future periods in the objective. We then use
  $g(\b'|M)$ to denote the expected total revenue the seller can collect from
  all upcoming periods by following a given bank account mechanism $M$ in these
  periods, where $\b'$ is the vector of the bank account balance at the end of
  the current period. In fact, $\b' = \b + \Delta \b(\v)$ is a function of
  buyers' bids $\v$ in this period, where $\Delta \b(\v)$ are the changes in the
  bank accounts.

  Then the expected revenue since current period is,
  \begin{align}\label{eq:obj}
    \Rev(\b, \xi)
      = \E\nolimits_\v\left[\sum\nolimits_i(v_i \cdot x_i(\v, \b) - \xi_i(v_{-i}))
                     + g(\Delta \b(\v) + \b | M)\right].
  \end{align}

  \paragraph{Constraints}
  Now we proceed to the constraints of the optimization. According to
  \autoref{lem:bank_dic}, the mechanism must satisfy constraint \eqref{eq:ic},
  \eqref{eq:bi}, and \eqref{eq:bu}. Note that \eqref{eq:bi} is guaranteed by the
  definition of $\xi_i(v_{-i})$, we then formalize constraint \eqref{eq:ic} and
  \eqref{eq:bu}.

  For \eqref{eq:ic}, as we won't introduce the payment variables, it is enough
  to require the monotonicity for the allocation rules only, i.e.,
  \begin{align}\label{eq:cic}
    \text{\eqref{eq:ic}}
      \iff x_i(\v, \b)~\text{is monotone in $v_i$ for fixed $v_{-i}$}.
  \end{align}

  Since by the Envelope theorem \citep{rochet1985taxation},
  \begin{align*}
    \frac{\partial (v_i \cdot x_i(\v, \b) - p_i(\v, \b))}{\partial v_i}
      = x_i(\v, \b)
    ~~ \Infer ~~
    u_i(\v, \b) = u_i(0, v_{-i}, \b) + \int_0^{v_i} x_i(s, v_{-i}, \b) \ud s.
  \end{align*}
  In other words, letting $p_i'(\v, \b) := p_i(\v, \b) - p_i(0, v_{-i}, \b)$ and
  $u_i'(\v, \b) := v_i \cdot x_i(\v, \b) - p_i'(\v, \b)$, we conclude that the
  utility of buyer $i$ in this period is fully determined by the allocation
  function $x_i(\cdot, \b)$ and the minimum payment $p_i(0, v_{-i}, \b)$:
  \begin{align*}
    \textstyle
    u_i'(\v, \b) = u_i(\v, \b) + p_i(0, v_{-i}, \b)
      = u_i(\v, \b) - u_i(0, v_{-i}, \b) + 0 \cdot x_i(0, v_{-i}, \b)
      = \int_0^{v_i} x_i(s, v_{-i}, \b) \ud s.
  \end{align*}

  Recall that for \eqref{eq:bu}, we require (i) the increase of balance cannot
  be more than the utility obtained in the current period and (ii) the updated
  balance must be nonnegative. Note that by the definition of $\xi_i(v_{-i})$
  \eqref{eq:xibi}:
  \begin{align*}
    \textstyle
    \E_{v_i}[u'_i(\v, \b) - p_i(0, v_{-i}, \b)]
    = \E_{v_i}[u_i(\v, \b)] = \xi_i(v_{-i}) ~~ \iff ~~
    p_i(0, v_{-i}, \b) = \E_{v_i}[u'_i(\v, \b)] - \xi_i(v_{-i}).
  \end{align*}

  Hence on one hand, for (i) and (ii), we can rewrite it as follows:
  \begin{align*}
    & \textstyle
      0 \leq b_i + \Delta b_i(\v) \leq b_i + u_i(\v, \b)  \\
    \iff ~& \textstyle
    -b_i \leq \Delta b_i(\v) \leq u_i(\v, \b)
      = u'_i(\v, \b) - p_i(0, v_{-i}, \b)
      = u'_i(\v, \b) - \E_{v_i}[u'_i(\v, \b)] + \xi_i(v_{-i}).
  \end{align*}

  On the other hand, to ensure the set of the possible balance increment
  $\Delta b_i(\v)$ satisfying (i) and (ii) is not empty, we need
  $b_i + u_i(\v, \b) \geq 0$. Note that $b_i$ is given and $u_i(\v, \b)$ reaches
  the minimum when $v_i = 0$, it is equivalent to:
  \begin{align*}
    \textstyle
    b_i + u_i(0, v_{-i}, \b) = b_i - p_i(0, v_{-i}, \b) \geq 0 ~~ \iff ~~
    \E_{v_i}[u'_i(\v, \b)] \leq b_i + \xi_i(v_{-i}).
  \end{align*}

  In summary,
  \begin{align}\label{eq:cbu}
    \text{\eqref{eq:bu}} \iff
      \left\{\begin{array}{l}
        -b_i \leq \Delta b_i(\v)
          \leq u'_i(\v, \b) - \E_{v_i}[u'_i(\v, \b)] + \xi_i(v_{-i})  \\
        \E_{v_i}[u'_i(\v, \b)] \leq b_i + \xi_i(v_{-i})
      \end{array}\right..
  \end{align}

  Finally, we also need the feasibility constraint of the allocation rule, that
  is, the total allocation of each item cannot be more than $1$:
  \begin{align}\label{eq:feas}
    \textstyle
    \sum_i x_i(\v, \b) \leq 1.
  \end{align}


  Therefore, combining the objective \eqref{eq:obj} and the constraints
  \eqref{eq:cic}, \eqref{eq:cbu}, and \eqref{eq:feas}, for given $\b$ and $g$,
  the optimization program for current period can be written as:
  \begin{align}
    \max \quad &
      \Rev(\b, \xi) =
        \E\nolimits_\v\left[\sum\nolimits_i (v_i \cdot x_i(\v) - \xi_i(v_{-i}))
                                 + g(\Delta \b(\v) + \b | M)\right]
      \label{eq:original}  \\
    \text{subj.t.} \quad &
      x_i(\v)~\text{is monotone in $v_i$ for any fixed $v_{-i}$} \nonumber  \\
      & \textstyle
        \E_{v_i}[u'_i(\v)] \leq b_i + \xi_i(v_{-i}),~\forall i, v_{-i}
        \nonumber  \\
      & - b_i \leq \Delta b_i(\v)
        \leq u'_i(\v) - \E_{v_i}[u'_i(\v)] + \xi_i(v_{-i}),
        ~\forall i,\v \nonumber  \\
      & \textstyle
        \sum_{i \in [n]} x_i(\v) \leq 1,~\forall \v \nonumber  \\
      & x_i(\v) \geq 0, \Delta b_i(\v) \geq 0,~\forall i,\v \nonumber
  \end{align}
  where we hide the $\b$ in the input to $x_i(\v, \b)$ and $u'_i(\v, \b)$ to
  emphasize that $\b$ is extraneous.

  \paragraph{Simplification}
  We then simplify the optimization program for optimal bank account mechanisms.
  One key observation is that with any fixed $\xi$ (for all periods), the
  optimization with any finite horizon can be solved via backward induction.

  In the program for the last period $T$, the expected revenue for the future,
  $g^T$, is always zero. For any parameter $\b$ of the last period, the optimal
  value of the program defines a function with respect to the balance vector
  $\b$, which, in fact, is the expected revenue for the last period where the
  mechanism in the last period $M^*_T$ is optimal (for given $\b$ and $\xi$).

  In other words, the optimal solution of the program for period $T$ defines the
  function $g^{T-1}(\b | M^*_T)$, where the mechanism for period $T$ is fixed to
  be optimal. Similarly, the optimal mechanism for each period $M^*_t$ and the
  corresponding $g^t(\b | M^*_{t+1..T})$ functions in the program for each
  period can be determined by backward induction. From now on, by omitting the
  mechanism being conditional on in $g^t$, we mean the $g^t$ function is
  conditional on the optimal $M^*_{t+1..T}$, i.e.,
  \begin{align*}
    g^t(\b) := g^t(\b | M^*_{t+1..T}) = \max_{M} g^t(\b | M).
  \end{align*}
  Note that the maximum always exists because the domain of the mechanisms is
  compact and the revenue is a continuous function of the mechanism.

  \begin{lemma}\label{lem:incg}
    $g^t(\b)$ is weakly increasing.
  \end{lemma}

  We place the proof of \autoref{lem:incg} in \autoref{ssec:prfg}.
  Hence we can eliminate variable $\Delta \b(\v)$ from the original program
  \eqref{eq:original}: Since $g^t(\b)$ is weakly increasing, it is without loss
  of generality to enforce
  \begin{align*}
    \textstyle
    \Delta b_i(\v) = u'_i(\v) - \E_{v_i}[u'_i(\v)] + \xi_i(v_{-i}),
    ~\forall i,\v.
  \end{align*}

  Meanwhile, by the following Myerson's lemma, we can further rewrite
  $\E_{v_i}[u'_i(\v)]$.
  \begin{lemma}[\cite{myerson1981optimal}]
    For any incentive compatible single item auction, the expected payment of
    buyer $i$ equals to the expected (Myerson's) virtual welfare contribution of
    buyer $i$:
    \begin{align*}
      \textstyle
      \forall i, v_{-i},~
        \E_{v_i}[p_i(\v)] = \E_{v_i}[(v_i - \vartheta_i(\v)) \cdot x_i(\v)],
    \end{align*}
    where $\vartheta_i(\v) = v_i \cdot (1 - F_i(v_i)) / f_i(v_i)$ is the value
    divided by the hazard rate.\footnote{$F_i$ must be absolutely integrable.}
  \end{lemma}

  Therefore,\footnote{$\vartheta_i(\v)$ is called {\em virtual value for utility}
  and equation \eqref{eq:vvu} is from \cite[Lemma 2.6]{HartlineMoneyBurning}.}
  \begin{align}\label{eq:vvu}
    \textstyle
    \E_{v_i}[u'_i(\v)] = \E_{v_i}[v_i \cdot x_i(\v) - p'_i(\v)]
      = \E_{v_i}[v_i \cdot x_i(\v) - (v_i - \vartheta_i(\v)) \cdot x_i(\v)]
      = \E_{v_i}[\vartheta_i(\v) \cdot x_i(\v)].
  \end{align}

  Thus, program \eqref{eq:original} can be simplified as follows, where
  $\Delta b_i(\v)$, $u'_i(\v)$, and $\bar u'_i(v_{-i})$ are notations, not
  variables.
  \begin{align}
    \max \quad &
      \Rev(\b, \xi) =
        \E\nolimits_\v\left[\sum\nolimits_i (v_i \cdot x_i(\v) - \xi_i(v_{-i}))
                          + g(\Delta \b(\v) + \b)\right]
      \label{eq:simple}  \\
    \text{subj.t.} \quad &
      x_i(\v)~\text{is monotone in $v_i$ for fixed $v_{-i}$} \nonumber  \\
      & \textstyle
        \bar u'_i(v_{-i}) := \E_{v_i}[u'_i(\v)] = \E_{v_i}[\vartheta_i(\v)x_i(\v)]
                      \leq b_i + \xi_i(v_{-i}),~\forall i, v_{-i} \nonumber  \\
      & \Delta b_i(\v) = u'_i(\v) - \bar u'_i(v_{-i}) + \xi_i(v_{-i}),
        ~\forall i,\v \nonumber  \\
      & \textstyle
        \sum_{i \in [n]} x_i(\v) \leq 1,~\forall \v \nonumber  \\
      & x_i(\v) \geq 0,~\forall i,\v \nonumber
  \end{align}
  %

  \paragraph{Convexity} We remains to show that the program is convex. In fact,
  we have the following lemma.
  \begin{lemma}\label{lem:cong}
    $g^t(\b)$ is and concave.
  \end{lemma}

  We place the proof of \autoref{lem:cong} in \autoref{ssec:prfg}. Since
  $g^t(\b)$ is concave, the objective function \eqref{eq:simple} is also
  concave. Meanwhile, all the constraints are linear, so the program is a convex
  program (note that a standard convex program is minimizing a convex objective
  function or maximizing a concave objective function).

  \subsection{Duality, virtual values, and ironing}

  We first consider the optimal solution to the program with the monotonicity
  constraint \eqref{eq:cic} on $\x(\v)$ relaxed. Later in \autoref{sec:ironing},
  we generalize the analysis to the original program. Although the optimal
  solution to the relaxed program is not a feasible solution, the analysis does
  capture the most interesting insight to this problem and provides the explicit
  form of the virtual values (before ironing). In fact, we show in
  \autoref{sec:ironing}, that adding the monotonicity constraint \eqref{eq:cic}
  back to the program corresponds to applying the ironing operation on the
  virtual values.

  %
  For the following relaxed program, let $\lambda_i(v_{-i})$ be the Lagrange
  multiplier of the constraint $\E_{v_i}[\vartheta_i(\v)x_i(\v)] \leq b_i +
  \xi_i(v_{-i})$ and $\mu(\v)$ be the Lagrange multiplier of the feasibility
  constraint \eqref{eq:feas}.
  \begin{align}
    \max \quad &
      \Rev(\b, \xi)
        = \E\nolimits_\v\left[\sum\nolimits_i (v_i \cdot x_i(\v) - \xi_i(v_{-i}))
                                 + g(\Delta \b(\v) + \b)\right]
               & \text{Lagrange multipliers}
      \label{eq:relaxed}  \\
    \text{subj.t.} \quad
      & \textstyle
        \bar u'_i(v_{-i}) := \E_{v_i}[u'_i(\v)] = \E_{v_i}[\vartheta_i(\v)x_i(\v)]
                      \leq b_i + \xi_i(v_{-i}),~\forall v_{-i}
               & \lambda_i(v_{-i}) \nonumber  \\
      & \Delta b_i(\v) = u'_i(\v) - \bar u'_i(v_{-i}) + \xi_i(v_{-i}),
        ~\forall i,\v
               & \varnothing  \nonumber  \\
      & \textstyle
        \sum_{i \in [n]} x_i(\v) \leq 1,~\forall \v
               & \mu(\v) \nonumber  \\
      & x_i(\v) \geq 0,~\forall i,\v \nonumber
  \end{align}
  Then the Lagrange $L := L(\x(\v), \lambda_i(v_{-i}), \mu(\v))$ is
  \begin{align*}
    L =
      & \textstyle
        \E_\v\left[\sum_i (v_i \cdot x_i(\v) - \xi_i(v_{-i}))
                    + g(\Delta \b(\v) + \b)\right]  \\
      & \textstyle
        - \sum_{i, v_{-i}} \lambda_i(v_{-i})
          \left(\E_{v_i}[\vartheta_i(\v)x_i(\v)] - b_i - \xi_i(v_{-i})\right)
        - \sum_{\v} \mu(\v) \left(\sum_i x_i(\v) - 1\right).
  \end{align*}


  Note that all constraints in program \eqref{eq:relaxed} are linear, therefore
  the Slater's condition is satisfied and the KKT conditions
  \citep[Chapter~5]{boyd2004convex} are necessary and
  sufficient for any optimal solution to \eqref{eq:relaxed}. In particular, the
  explicit form of the virtual values can be derived from the KKT conditions.

  Let $g_i$ denote the partial derivative of $g$ with respect to the $i$-th
  dimension, i.e., $g_i(\b) = \partial g(\b) / \partial b_i$. Let $\alpha$ and
  $\beta$ be defined as follows,
  \begin{align*}
    \textstyle
    \alpha_i(\v) := 1 + g_i(\Delta \b(\v) + \b) \qquad
    \beta_i(v_{-i})
      := \frac{\lambda_i(v_{-i})}{f(v_{-i})} + \E_{v_i}[g_i(\Delta \b(\v) + \b)].
  \end{align*}

  \begin{lemma}[Virtual welfare maximizer]\label{lem:virtual}
    Any optimal solution to \eqref{eq:relaxed} must maximize the expected
    virtual welfare, where the virtual value function $\phi_i(\v)$ is given as
    follows:
    \begin{align*}
      \phi_i(\v)
        = \alpha_i(\v) \cdot v_i - \beta_i(v_{-i}) \cdot \vartheta_i(\v).
    \end{align*}
  \end{lemma}

  We place the detailed proof of this lemma in \autoref{ssec:virtual}.

  \paragraph{Ironing}
  As we mentioned at the beginning of this subsection, the optimal solution to
  the original program \eqref{eq:simple} must maximize the expected ironed
  virtual welfare defined according to the ironed virtual values
  $\tilde \phi_i(\v)$. Moreover, the ironed virtual value $\tilde \phi_i(\v)$ is
  the transformation of $\phi_i(\v)$ according to the ironing rules
  (\autoref{def:ironing}). We omit the details of the analysis here and
  incorporate the conclusion in the formal version of our first main theorem
  (\autoref{thm:formal}). We refer the readers to \autoref{sec:ironing} for the
  complete proof.

  \subsection{Sensitivity analysis and the optimality across the periods}

  So far we have determined the optimal auction for each period when the
  expected buyer utilities of each period $\xi$ are fixed. Now we show how to
  optimize the remaining parameters $\xi$ through sensitivity analysis. Let
  $R(\xi)$ denote the optimal revenue when the auctions in each period are
  optimal for the given $\xi$. Then we show how to determine the partial
  derivatives of $R(\xi)$ via sensitivity analysis and hence enable the gradient
  descent algorithm to find the optimal $\xi^*$. In particular, $R(\xi)$ is a
  concave function.\footnote{Since by simply taking any convex combination of
  different $\xi$ and $\xi'$, the revenue obtained is also the convex
  combination of the revenue resulted by $\xi$ and $\xi'$.}

  The standard sensitivity analysis \cite[Chapter 5.6]{boyd2004convex} indicates
  how much the optimal objective value of a program will change if some
  constraints become slightly looser or tighter. Such quantities usually have
  important physical meanings in economic setups. For example, consider
  $g^t(\b) = \max \Rev^t(\b, \xi)$ (we brought back the superscript ${}^t$ to
  distinguish the variables for different periods). By sensitivity
  analysis,\footnote{If $g^t$ is not differentiable at $\b$, the right-hand-side
  is a subgradient of $g^t$.} we have:
  \begin{align*}
    \textstyle
    g^t_i(\b) = \E_{\v^t}[g^{t+1}_i(\Delta \b(\v) + \b)]
                + \sum_{v^t_{-i}}\lambda^t_i(v^t_{-i}).
  \end{align*}
  Then $g^t_i(\b)$, by definition, is the marginal contribution of the balance
  of buyer $i$ in period $t$ to the expected revenue since the $t$-th period.
  Similarly, $\E_{\v^t}[g^{t+1}_i(\Delta \b(\v) + \b)]$ is the marginal
  contribution of $b_i$ to the expected revenue since the $(t+1)$-th period.
  Hence their difference, $\sum_{v^t_{-i}}\lambda^t_i(v^t_{-i})$, is the
  marginal contribution of $b_i$ to the expected revenue of the $t$-th period.
  In particular, $\lambda^t_i(v^t_{-i})$ is the marginal contribution if the
  values of other buyers are $v^t_{-i}$.

  Moreover, we can conclude the concrete form of the partial derivatives of
  $R(\xi)$ by sensitivity analysis:
  \begin{align*}
    & \textstyle
    \frac{\partial g^t(\b)}{\partial \xi^t_i(v^t_{-i})} = -f^t(v^t_{-i})
      + f^t(v^t_{-i}) \cdot \E_{v^t_i}[g^{t+1}_i(\Delta \b(\v) + \b)]
      + \lambda^t_i(v^t_{-i}) = (\beta^t_i(v^t_{-i}) - 1)f^t(v^t_{-i})  \\
    \Infer~& \textstyle
    \frac{\partial R(\xi)}{\partial \xi^t_i(v^t_{-i})}
      = \frac{\partial \E_{\b}\left[g^t(\b)\right]}{\partial \xi^t_i(v^t_{-i})}
      = (\E_{\b}\left[\beta^t_i(v^t_{-i})\right] - 1)f^t(v^t_{-i}),
  \end{align*}
  where the expectation taken over $\b$ is computed by simulating the auctions
  in previous periods. In particular, $\xi$ is selected optimally, if and only
  if:
  \begin{align*}
    \textstyle
    \E_{\b}\left[\beta^t_i(v^t_{-i})\right] = 1 \quad \text{or} \quad
    \xi^t(v^t_{-i}) = 0,~
    \E_{\b}\left[\beta^t_i(v^t_{-i})\right] \leq 1.
  \end{align*}
  %
  %
  %


  \subsection{Summary}

  As a summary, we formally restate \autoref{thm:informal}.
  \begin{theorem}[Formal]\label{thm:formal}
    A bank account mechanism is optimal if and only if all the following
    conditions are satisfied:
    \begin{itemize}
      \itemsep0em
      \item it satisfies all the basic constraints, \eqref{eq:ic},
            \eqref{eq:bi}, \eqref{eq:bu} and the feasibility constraint
            \eqref{eq:feas};
      \item its allocation rule maximizes the ironed virtual welfare, where the
            virtual value (before ironing) $\phi(\v)$ has the following form:
            \begin{align*}
              \phi(\v) = \alpha_i(\v)v_i - \beta_i(v_{-i})\vartheta_i(\v),
            \end{align*}
            which can be seen as the combination of the Myerson's virtual value
            and the private value;
      \item finally, the expected utility of each period $\xi$ is selected
            optimally, i.e.,
            \begin{align*}
              \textstyle
              \E_{\b}\left[\beta^t_i(v^t_{-i})\right] = 1
              \quad \text{or} \quad
              \xi^t(v^t_{-i}) = 0,~\E_{\b}\left[\beta^t_i(v^t_{-i})\right] \leq 1.
            \end{align*}
    \end{itemize}
  \end{theorem}


\section{An Algorithmic Approach to the Structure}\label{sec:fptas}

  So far we showed that the optimal bank account mechanism maximizes the ironed
  virtual welfare in each period. Although the explicit form of the virtual
  values (before ironing) is given in \autoref{lem:virtual} and the ironing rule
  is given in \autoref{def:ironing}, how to accomplish the ironing operation is
  still unknown. One of the major difficulty of the ironing comes from the fact
  that one's virtual value not only depends on his/her own private value, but
  also depends on the private values of other buyers through the term
  $\alpha_i(\v)$. In the presence of such interdependence across different
  buyers on virtual values, monotone virtual value function does not imply
  monotone allocation rules.

  In this section, we algorithmically resolve the difficulty of ironing. In
  particular, we show a Fully Polynomial Time Approximation Scheme (FPTAS)
  that can accomplish the ironing step for any constant many buyer cases and
  hence compute the ironed virtual values. Moreover, the bank account mechanism
  induced by the ironed virtual values computed is (multiplicatively)
  $(1 - \epsilon)$-approximately optimal in terms of revenue.

  We also emphasize that the hard core of computing the exact optimal solution
  is not directly from the ironing step but the step of approximating the
  concave functions $g^t(\b)$. Note that $g^t(\b)$ is a continuous function
  without closed forms and the main effort of this section is to show that we
  can arbitrarily approximate $g^t(\b)$ with piece-wise linear functions and
  guarantee that (i) the number of pieces is polynomial in the input size and
  (ii) the final result is approximately optimal.

  Recall the program \eqref{eq:unrelaxed}, in particular, we bring back the
  superscripts ${}^t$ to emphasize the periods:
  \begin{align*}
    \max \quad &
      g^{t-1}(\b^t) = \Rev(\b^t, \xi^t) =
        \E\nolimits_{\v^t}\left[\sum\nolimits_i (v^t_i \cdot x^t_i(\v^t) - \xi^t_i(v^t_{-i}))
          + g^t(\Delta \b^t(\v^t) + \b^t)\right]  \\
    \text{subj.t.} \quad &
        x^t_i(v^t_i, v^t_{-i}) - x^t_i({v^t_i}', v^t_{-i}) \leq 0,
        ~\forall i,v^t_{-i},{v^t_i}' > v^t_i  \\
      & \textstyle
        \bar u^t_i(v^t_{-i}) := \E_{v^t_i}[{u^t_i}']
          = \E_{v^t_i}[\vartheta^t_i(\v^t)x^t_i(\v^t)]
          \leq b^t_i + \xi^t_i(v^t_{-i}),~\forall v^t_{-i}  \\
      & \Delta b^t_i(\v^t)
          = {u^t_i}'(\v^t) - \bar u^t_i(v^t_{-i}) + \xi^t_i(v^t_{-i}),
        ~\forall i,\v^t  \\
      & \textstyle
        \sum_{i \in [n]} x^t_i(\v^t) \leq 1,~\forall \v^t  \\
      & x^t_i(\v^t) \geq 0,~\forall i,\v^t
  \end{align*}

  In what follows, we formalize the FPTAS via dynamic programming to compute the
  optimal ironed virtual values for the discrete type case.
  \begin{theorem}\label{thm:dp}
    The ironed virtual values of the optimal bank account mechanism can be
    computed through a dynamic programming based algorithm.

    Moreover, for any $\epsilon > 0$, there is an FPTAS to achieve an
    $\epsilon$-approximation (multiplicative) of the optimal revenue.
  \end{theorem}

  We first outline the main idea of the algorithm: (i) for any fixed $\xi$,
  compute the ironed virtual values of the approximately optimal bank account
  mechanism; (ii) compute the optimal $\xi$ using gradient descent. Since we
  have shown that the revenue of the bank account mechanism is a concave
  function with respect to $\xi$, the second step is standard and can be done
  with polynomially many queries to the (approximately) optimal revenue as a
  function of $\xi$. In what follows, we will focus on the first step.

  To compute the ironed virtual value, we need to solve the dual program of
  \eqref{eq:unrelaxed}, which, of course, is equivalent to solve the primal
  because the strong duality holds. Note that when $\xi$ and $\b^t$ are fixed,
  \eqref{eq:unrelaxed} is a standard convex program with polynomially many
  linear constraints. Hence the optimal solution to its dual could be computed
  efficiently given oracle accesses to the concave function $g^t$. Then by the
  definition of $g^{t-1}$, the value of the optimal solution is $g^{t-1}(\b^t)$.
  Therefore, the concave function $g^t$ can be evaluated recursively for each
  $t$ and hence the ironed virtual values of the optimal bank account mechanism
  can be computed via standard dynamic programming as well.

  However, the computation of the entire dynamic programming is not directly
  efficient: (i) if each $g^t$ is computed recursively upon every query, the
  depth of the recursion could be up to $T$ and hence the total number of
  queries required would be exponential in $T$; (ii) if each $g^t$ is computed
  once for all possible $\b$ so that any further queries of $g^t$ can be
  answered from precalculated values, then the number of input points where
  $g^t$ need to be computed would be unbounded.

  The key step to resolve these issues is to approximate the concave functions
  $g^t(\b)$ by piece-wise linear functions each with at most polynomially many
  pieces. In addition, each of the piece-wise linear functions can be further
  expressed as the minimum of a set of affine functions. Hence both the original
  convex program \eqref{eq:unrelaxed} at each period $t$ and its dual can be
  approximated by a polynomially large linear program.

  The following lemma ensures that each $g^t$ can be well approximated by a
  piece-wise linear function.
  \begin{lemma}\label{lem:pwlinear}
    For all $t \in [T]$, $g^t$ can be $\kappa$-approximated by two concave
    piece-wise linear functions, $\ubar g^t$ and $\bar g^t$:
    \begin{align*}
      \textstyle
      \forall \b,~\ubar g^t(\b) \leq g^t(\b) \leq \bar g^t(\b)
        \leq \ubar g^t(\b) + \kappa \max_{\b'} g^t(\b').
    \end{align*}
    Moreover, each of $\ubar g^t$ and $\bar g^t$ can be written as the minimum
    of at most polynomially many pieces.
  \end{lemma}

  With such approximations to the concave function $g^t$, we can approximately
  solve the convex program \eqref{eq:unrelaxed} by solving the following linear
  program, where each occurrence of $g^t$ is replaced by $\bar g^t$:
  \begin{align}
    \max \quad &
      \bar h^{t-1}(\b^t) :=
        \E\nolimits_{\v^t}\left[\sum\nolimits_i (v^t_i \cdot x^t_i(\v^t) - \xi^t_i(v^t_{-i}))
          + \bar g^t(\Delta \b^t(\v^t) + \b^t)\right]
      \label{eq:lpbar}  \\
    \text{subj.t.} \quad &
        x^t_i(v^t_i, v^t_{-i}) - x^t_i({v^t_i}', v^t_{-i}) \leq 0,
        ~\forall i,v^t_{-i},{v^t_i}' > v^t_i \nonumber  \\
      & \textstyle
        \bar u^t_i(v^t_{-i}) := \E_{v^t_i}[{u^t_i}']
          = \E_{v^t_i}[\vartheta^t_i(\v^t)x^t_i(\v^t)]
          \leq b^t_i + \xi^t_i(v^t_{-i}),~\forall v^t_{-i} \nonumber  \\
      & \Delta b^t_i(\v^t)
          = {u^t_i}'(\v^t) - \bar u^t_i(v^t_{-i}) + \xi^t_i(v^t_{-i}),
        ~\forall i,\v^t \nonumber  \\
      & \textstyle
        \sum_{i \in [n]} x^t_i(\v^t) \leq 1,~\forall \v^t \nonumber  \\
      & x^t_i(\v^t) \geq 0,~\forall i,\v^t \nonumber  \\
      & \bar g^t(\Delta \b^t(\v^t) + \b^t)
          \leq \boldsymbol{\alpha}_l \cdot (\Delta \b^t(\v^t) + \b^t) + \beta_l,
          ~ \forall l,\v^t \label{eq:minaffine}
  \end{align}

  As we mentioned previously, in the last constraint \eqref{eq:minaffine}, we
  assume that the function $\bar g^t$ can be expressed by the minimum of a set
  of affine functions, i.e.,
  \begin{align*}
    \textstyle
    \bar g^t(\b) = \min_{l \in L} \boldsymbol{\alpha}_l \cdot \b + \beta_l.
  \end{align*}

  Here we slightly abuse the notation of $\bar g^t(\Delta \b^t(\v^t) + \b^t)$ as
  variables in the linear program. Note that including the constraints
  \eqref{eq:minaffine} in the linear program \eqref{eq:lpbar} suffices to ensure
  that the variable values always agree with the corresponding function
  values. Because on the one hand, by the constraints, each variable is no more
  than the corresponding affine functions; on the other hand, since the
  coefficients of these variables are always positive in the objective, for each
  of the variables, at least one of the constraints in \eqref{eq:minaffine} must
  be binding.

  Let $\bar h^{t-1}(\b)$ denote the optimal value of the linear program
  \eqref{eq:lpbar} when $\b^t = \b$. Similarly, we can define $\ubar h^{t-1}$ by
  simply replacing all $\bar g^t$ with $\ubar g^t$ in \eqref{eq:lpbar}. The
  following lemma shows that $\ubar h^{t-1}$ and $\bar h^{t-1}$ are in fact
  lower and upper bounds of $g^{t-1}$.
  \begin{lemma}\label{lem:pwapprox}
    $\ubar h^{t-1}$ and $\bar h^{t-1}$ are concave and
    \begin{align*}
      \forall \b,~\ubar h^{t-1}(\b) \leq g^{t-1}(\b) \leq \bar h^{t-1}(\b)
        \leq \ubar h^{t-1}(\b) + \max_{\b'} (\bar g^t(\b') - \ubar g^t(\b')).
    \end{align*}
  \end{lemma}

  \begin{proof}[Proof of \autoref{lem:pwapprox}]
    Let $\ubar \x^*$, $\x^*$, and $\bar \x^*$ be the corresponding optimal
    solution for $\ubar h^{t-1}(\b)$, $g^{t-1}(\b)$, and $\bar h^{t-1}(\b)$,
    respectively. Denote these three programs as $\ubar P$, $P$, and $\bar P$,
    respectively. Note that $\ubar \x^*$ is feasible in $P$, hence by the
    optimality of $\x^*$,
    \begin{align*}
      P(\ubar \x^*) \leq P(\x^*),
    \end{align*}
    where $P(\x)$ denotes the objective value of program $P$ with variables
    being $\x$.

    On the other hand, for any $\x$ that is feasible to both $\ubar P$ and $P$,
    we have,
    \begin{align*}
      \ubar P(\x) \leq P(\x).
    \end{align*}
    Because in the objectives, $\ubar g^{t-1} \leq g^{t-1}$. Hence we conclude that
    \begin{align*}
      \ubar h^{t-1}(\b) = \ubar P(\ubar \x^*)
        \leq P(\ubar \x^*) \leq P(\x^*) = g^{t-1}(\b).
    \end{align*}

    Similarly, we can prove that $g^{t-1}(\b) \leq \bar h^{t-1}(\b)$. For the
    last inequality, denote $\delta = \max_{\b'}(\bar g^t(\b')-\ubar g^t(\b'))$.
    Note that for the $\bar g$ variables in $\bar \x^*$, if we reduce all of
    them by $\delta$ to get $\bar \x^{**}$, $\bar \x^{**}$ must be feasible to
    $\ubar P$ and the objective value is reduced by at most $\delta$, hence
    \begin{align*}
      \ubar P(\x^*) \geq \ubar P(\bar \x^{**}) \geq \bar P(\bar \x^*) - \delta.
    \end{align*}

    In other words,
    \begin{align*}
      \ubar h^{t-1}(\b) \geq \bar h^{t-1}(\b) - \delta.
    \end{align*}
  \end{proof}

  Even through both $\ubar h^{t-1}$ and $\bar h^{t-1}$ are in fact piece-wise
  linear functions, but they cannot be directly used for the computation of
  period $t-1$, because they may have exponentially many pieces. However, since
  they are concave, we can apply \autoref{lem:pwlinear} to get the lower bound
  of $\ubar h^{t-1}$ and the upper bound of $\bar h^{t-1}$:
  \begin{align*}
    \ubar g^{t-1}(\b) \leq \ubar h^{t-1}(\b) \leq g^{t-1}(\b)
      \leq \bar h^{t-1}(\b) \leq \bar g^{t-1}(\b).
  \end{align*}

  Therefore, we can recursively compute $\ubar g^1, \ldots, \ubar g^T$ and
  $\bar g^1, \ldots, \bar g^T$ and hence compute the ironed virtual values of
  the approximately optimal bank account mechanism for any fixed $\xi$.
  Combining with the fact that $\xi$ could then be optimized using gradient
  descent, we are done with our algorithm.

\bibliographystyle{named}
\bibliography{charact}
\clearpage

\begin{appendix}
  \section{Missing Proofs}

\subsection{Proof of \autoref{lem:incg} and \autoref{lem:cong}}\label{ssec:prfg}

\begin{proof}[Proof of \autoref{lem:incg} and \autoref{lem:cong}]
  We prove these lemmas by induction from $t = T$ down to $1$. Since $g^T \equiv
  0$, all these properties are satisfied for $t = T$.

  Suppose that $g^{t + 1}(\b)$ is weakly increasing and concave. We first show
  that $g^t(\b)$ is weakly increasing. Consider $\b' \geq \b$. By definition,
  $g^t(\b)$ is the optimal value of program \eqref{eq:original} for period $t$
  (with given $\xi$ and $\b$). Note that:
  \begin{enumerate}
    \item The optimal solution for program \eqref{eq:original} with balance
          $\b$ is also feasible for the program with balance $\b'$;
    \item The objective is weakly increasing in $\b$ even with the same
          variables. The first part of the objective, $\E_\v[\sum_i(v_i \cdot
          x_i(\v) - \xi_i(v_{-i}))]$, is the same for the same allocation rule
          and the second part of the objective, $\E[g^{t+1}(\Delta \b(\v) +
          \b)]$, is weakly increasing in $\b$ by induction (with the same
          $\Delta \b(\v)$).
  \end{enumerate}
  Therefore, $g^t(\b') \geq g^t(\b)$ is weakly increasing.

  We then show that $g^t(\b)$ is concave, i.e., for any $\theta \in [0, 1]$,
  \begin{align*}
    g^t(\theta \b + (1 - \theta) \b')
      \geq \theta g^t(\b) + (1 - \theta) g^t(\b').
  \end{align*}
  Let $\x^*$ and $\Delta \b^*$ be the optimal solution for balance vector $\b$
  and $\x^{**}$ and $\Delta \b^{**}$ be the optimal solution for balance
  vector $\b'$. Then their convex combination, $(\theta \x^* + (1 - \theta)
  x^{**}, \theta \Delta \b^* + (1 - \theta) \Delta \b^{**})$, is a feasible
  solution for balance vector $\theta \b + (1 - \theta) \b'$. Meanwhile, the
  objective function is concave in the variables,
  \begin{align*}
    & g^t(\theta \b + (1 - \theta) \b')  \\
    \geq~&
    \E_v\left[\sum_i(v_i \cdot (\theta x^*_i(\v) + (1 - \theta) x^{**}_i(\v))
      - \xi_i(v_{-i})) + g^{t+1}(\theta(\Delta \b^*(\v) + \b)
                                 + (1 - \theta)(\Delta \b^{**}(\v) + \b'))
        \right]  \\
    \geq~& \theta \E_v\left[\sum_i(v_i \cdot x^*_i(\v) - \xi_i(v_{-i}))
                            + g^{t+1}(\Delta \b^*(\v) + \b) \right]  \\
         & + (1 - \theta)
            \E_v\left[\sum_i(v_i \cdot x^{**}_i(\v) - \xi_i(v_{-i}))
                      + g^{t+1}(\Delta \b^{**}(\v) + \b') \right]  \\
    =~& \theta g^t(\b) + (1 - \theta) g^t(\b'),
  \end{align*}
  where the first inequality is by the definition of $g^t(\b)$ and the
  feasibility of the convex combination we just argued and the second
  inequality is by the induction of the concavity of $g^{t+1}(\b)$.
\end{proof}

\subsection{Proof of \autoref{lem:virtual}}\label{ssec:virtual}

  \begin{proof}[Proof of \autoref{lem:virtual}]
    By KKT conditions, $\partial L / \partial x_i(\v) \geq 0$:\footnote{If not
    differentiable, then any sub-gradient must be non-negative.}
    \begin{align*}
      \frac{\partial L}{\partial x_i(\v)} & = \left(
        v_i - \frac{\lambda_i(v_{-i}) \vartheta_i(\v)}{f(v_{-i})}\right) f(\v)
        + \frac{\partial \E_{\v}[g(\Delta \b(\v) + \b)]}{\partial x_i(\v)}
        - \mu(\v)  \\
      & = \left(\left(1 + g_i(\Delta \b(\v) + \b)\right)v_i
                - \left(\frac{\lambda_i(v_{-i})}{f(v_{-i})}
                + \E_{v_i}[g_i(\Delta \b(\v) + \b)]\right)
                \vartheta_i(\v)\right) f(\v) - \mu(\v)  \\
      & = (\alpha_i(\v) v_i - \beta_i(v_{-i}) \vartheta_i(\v)) f(\v) - \mu(\v),
    \end{align*}
    where
    \begin{align*}
      \frac{\partial \E_{\v}[g(\Delta \b(\v) + \b)]}{\partial x_i(\v)}
      & = \frac{\partial \E_{\v}[g(\Delta \b(\v) + \b)]}{\partial u'_i(\v)}
          \cdot \frac{\partial u'_i(\v)}{\partial x_i(\v)}
        + \frac{\partial \E_{\v}[g(\Delta \b(\v) + \b)]}{\partial \bar u'_i(v_{-i})}
          \cdot \frac{\partial \bar u'_i(v_{-i})}{\partial x_i(\v)}  \\
      & = g_i(\Delta \b(\v) + \b)f(\v) v_i
          - \E_{v_i}[g_i(\Delta \b(\v) + \b)]f(v_{-i}) \cdot f(v_i)\vartheta_i(\v)
    \end{align*}
    and if $x_i(\v) > 0$, $\partial L / \partial x_i(\v)$ must be $0$, i.e.,
    \begin{align*}
      & \forall i,~\v,~ (\alpha_i(\v) v_i - \beta_i(v_{-i}) \vartheta_i(\v)) f(\v)
                          - \mu(\v) \leq 0  \\
      & x_i(\v) > 0 \Infer
        (\alpha_i(\v) v_i - \beta_i(v_{-i}) \vartheta_i(\v)) f(\v) - \mu(\v) = 0
    \end{align*}
    Then for any $x_i(\v) > 0$, we have
    \begin{align*}
      \alpha_i(\v) v_i - \beta_i(v_{-i}) \vartheta_i(\v) = \mu(\v) / f(\v)
        \geq \alpha_j(\v) v_j - \beta_j(v_{-j}) \vartheta_j(\v).
    \end{align*}

    In other words, the optimal solution is a virtual value maximizer, where the
    virtual value $\phi_i(\v)$ is given as
    \begin{align*}
      \phi_i(\v) = \alpha_i(\v) v_i - \beta_i(v_{-i}) \vartheta_i(\v).
    \end{align*}

    In particular, $\lambda_i(v_{-i}) = 0$ when
    \begin{align*}
      \E_{v_i}[\vartheta_i(\v)x_i(\v)] - b_i - \xi_i(v_{-i}) < 0,
    \end{align*}
    i.e., the balance of buyer $i$ is not spent out.
  \end{proof}

\subsection{Proof of \autoref{lem:pwlinear}}\label{ssec:prfpwlinear}

  \begin{proof}[Proof of \autoref{lem:pwlinear}]
    The main idea of the proof is to (i) select polynomially many different
    points $\b$ in the $k$-dimension space $\R^k_+$, (ii) compute $g(\b)$ for
    the selected points, and (iii) use the upper surface of the convex hull of
    $\{(\b, g(\b))\}$ as $\ubar h(\b)$.

    Note that step (ii) and (iii) can be done very efficiently.\footnote{For
    (ii), the computation is efficient by hypothesis. For (iii), the convex hull
    in $(k + 1)$-dimension can be using the quickhull algorithm
    \citep{barber1996quickhull}. Note that $k$ is a constant here.} However, the
    step (i) is non-trivial.

    To illustrate the algorithm, we first highlight that both the maximum value
    of each $b_i$, $g(\b)$, and the partial derivatives $g_i(\b)$ are bounded by
    $2^{O(N)}$, where $N$ is the input size. In other words, simply selects
    points with fixed distance between other nearby points will result in
    exponentially many points. Therefore, the distances should be determined
    adaptively.

    The key observation here is to keep track of an upper bound on the
    uncertainty of the partial derivative along each direction and the domain of
    $b_i$ where the difference of $g(\b)$ and $\ubar h(\b)$ could be more than
    $\kappa \cdot \max_{\b'} g(\b')$. Then each time pick a point on the
    boundary of the domain and the upper of the uncertainty of the partial
    derivative will be reduced by at least half. Hence after polynomially many
    steps, the uncertainty of the partial derivative will be exponentially small
    and we can easily select modest numbers of points to get a good enough
    approximation $\ubar h(\b)$ to the target function $g(\b)$.
  \end{proof}

\section{Ironing}\label{sec:ironing}

  We add the monotonicity constraint \eqref{eq:cic} back and show that this
  constraint corresponds to applying ironing operation on the virtual value
  functions.

  Consider the following program without relaxation, where the monotonicity
  constraint is rewritten in an explicit way.
  \begin{align}
    \max \quad &
      \Rev(\b, \xi) =
        \E_\v\left[\sum\nolimits_i (v_i \cdot x_i(\v) - \xi_i(v_{-i}))
                                 + g(\Delta \b(\v) + \b)\right]
               & \text{Lagrange multipliers}
      \label{eq:unrelaxed}  \\
    \text{subj.t.} \quad &
        x_i(v_i, v_{-i}) - x_i(v'_i, v_{-i}) \leq 0,
        ~\forall i,v_{-i},v'_i > v_i
               & \eta_i(v_i, v'_i; v_{-i}) \nonumber  \\
      & \textstyle
        \bar u'_i(v_{-i}) \leq b_i + \xi_i(v_{-i}),~\forall i, v_{-i}
               & \lambda_i(v_{-i}) \nonumber  \\
      & \Delta b_i(\v) = u'_i(\v) - \bar u'_i(v_{-i}) + \xi_i(v_{-i}),
        ~\forall i,\v
               & \varnothing  \nonumber  \\
      & \textstyle
        \sum_{i \in [n]} x_i(\v) \leq 1,~\forall \v
               & \mu(\v) \nonumber  \\
      & x_i(\v) \geq 0,~\forall i,\v \nonumber
  \end{align}

  Then the Lagrange for this program is as follows.
  \begin{align}
    &~L(\x, \boldeta, \boldlambda, \boldmu) \nonumber  \\
    = & \E_\v\left[\sum_i (v_i \cdot x_i(\v) - \xi_i(v_{-i}))
                    + g(\Delta \b(\v) + \b)\right] \nonumber  \\
      & - \sum_i \sum_\v x_i(\v)
          \left(\sum_{v'_i > v_i}\eta_i(v_i, v'_i; v_{-i})
                - \sum_{v''_i < v_i}\eta_i(v''_i, v_i; v_{-i})\right)
          \nonumber  \\
      & - \sum_{i, v_{-i}} \lambda_i(v_{-i})
          \left(\E_{v_i}[\vartheta_i(\v)x_i(\v)] - b_i - \xi_i(v_{-i})\right)
        - \sum_{\v} \mu(\v) \left(\sum_i x_i(\v) - 1\right).
        \label{eq:lagrange} \tag{\text{L}}
  \end{align}
  By similar argument with the proof of \autoref{lem:virtual}, any optimal
  auction must maximize the ironed virtual welfare and the {\em ironed virtual
  value function} $\tilde \phi_i(\v)$ is given as
  \begin{align*}
    \tilde \phi_i(\v) = \alpha_i(\v) v_i - \beta_i(v_{-i}) \vartheta_i(\v)
      - \frac1{f(\v)}\left(\sum_{v'_i > v_i}\eta_i(v_i, v'_i; v_{-i})
            - \sum_{v''_i < v_i}\eta_i(v''_i, v_i; v_{-i})\right).
  \end{align*}

  Note that conditional on any $v_{-i}$, the expectation of the virtual value
  with ironing equals to the expectation of the virtual value without ironing:
  \begin{align*}
    \sum_{v_i} f_i(v_i) \phi_i(\v) = \sum_{v_i} f_i(v_i) \tilde \phi_i(\v).
  \end{align*}
  Thus $\eta_i(v_i, v'_i; v_{-i}) / f(v_{-i})$ defines a mass move of
  $f_i(v_i)\phi_i(\v)$ within $\R_+$. In particular, since
  $\eta_i(v_i, v'_i; v_{-i}) \geq 0$, the move is from small $v_i$ to large
  $v'_i$ ($v_i < v'_i$).

  Meanwhile, by complementary slackness,
  \begin{align*}
    (x_i(v_i, v_{-i}) - x_i(v'_i, v_{-i}))\eta_i(v_i, v'_i; v_{-i}) = 0,
  \end{align*}
  which means the mass move must happen within the regions with the same
  allocation.

  Also note that the virtual value $\phi_i(\v)$ depend on $v_{-i}$ as well, the
  ironing does not make $\tilde \phi_i(\v)$ an increasing function. Instead, the
  following function must be weakly increasing,
  \begin{align*}
    \sign\left[\tilde \phi_i(\v)
               - \max\big\{0, \max_{j \neq i} \tilde \phi_j(\v)\big\}\right].
  \end{align*}
  Because the allocation maximizes the virtual welfare and must be weakly
  increasing.

\end{appendix}
\end{document}